\documentclass[conference]{IEEEtran}
\IEEEoverridecommandlockouts
\usepackage{times}
\usepackage{epsfig}
\usepackage{amsmath}
\usepackage{amssymb}
\usepackage{subcaption} 
\usepackage{overpic}
\usepackage{cite}
\usepackage{amsmath,amssymb,amsfonts}
\usepackage{amsthm}
\usepackage{algorithmic,algorithm}
\usepackage{graphicx}
\usepackage{textcomp}
\usepackage{xcolor}
\usepackage{mathrsfs}
\usepackage{url}
\usepackage{color}
\usepackage{booktabs}
\usepackage{tabularx}
\newtheorem{theorem}{Theorem}
\usepackage{booktabs} 
\usepackage{makecell} 
\usepackage{multirow} 
\usepackage{tablefootnote} 

\def\BibTeX{{\rm B\kern-.05em{\sc i\kern-.025em b}\kern-.08em
    T\kern-.1667em\lower.7ex\hbox{E}\kern-.125emX}}
\begin{document}

\title{ Joint Optimization of Prompt Security and System Performance in Edge-Cloud LLM Systems
\thanks{The research was supported 
in part by the Chinese National Research Fund (NSFC) under Grant 62272050 and Grant 62302048; in part by the Guangdong Key Lab of AI and Multi-modal Data Processing, United International College (UIC), Zhuhai under 2023-2024 Grants sponsored by Guangdong Provincial Department of Education; in part by Institute of Artificial Intelligence and Future Networks (BNU-Zhuhai) and Engineering Center of AI and Future Education, Guangdong Provincial Department of Science and Technology, China; Zhuhai Science-Tech Innovation Bureau under Grant No. 2320004002772; in part by Guangdong Basic and Applied Basic Research Fund 2021A1515110190 and in part by the Interdisciplinary Intelligence SuperComputer Center of Beijing Normal University (Zhuhai). 
Corresponding authors: Tianhui Meng (tianhuimeng@uic.edu.cn) and Weijia Jia (jiawj@bnu.edu.cn).}
}
\author{
	\IEEEauthorblockN{Haiyang Huang\IEEEauthorrefmark{1}, Tianhui Meng\IEEEauthorrefmark{2} and Weijia Jia\IEEEauthorrefmark{1}\IEEEauthorrefmark{2}} 
	\IEEEauthorblockA{\IEEEauthorrefmark{1}Institute of Artificial Intelligence and Future Networks, Beijing Normal University, Zhuhai, China} 
    \IEEEauthorblockA{\IEEEauthorrefmark{2}Department of Computer Science, BNU-HKBU United International College, Zhuhai, China}
}

\maketitle

\begin{abstract}
Large language models (LLMs) have significantly facilitated human life, and prompt engineering has improved the efficiency of these models. However, recent years have witnessed a rise in prompt engineering-empowered attacks, leading to issues such as privacy leaks, increased latency, and system resource wastage. Though safety fine-tuning based methods with Reinforcement Learning from Human Feedback (RLHF) are proposed to align the LLMs, existing security mechanisms fail to cope with fickle prompt attacks, highlighting the necessity of performing security detection on prompts.
In this paper, we jointly consider prompt security, service latency, and system resource optimization in Edge-Cloud LLM (EC-LLM) systems under various prompt attacks. 
To enhance prompt security, a vector-database-enabled lightweight attack detector is proposed. 
We formalize the problem of joint prompt detection, latency, and resource optimization into a multi-stage dynamic Bayesian game model. The equilibrium strategy is determined by predicting the number of malicious tasks and updating beliefs at each stage through Bayesian updates. 
The proposed scheme is evaluated on a real implemented EC-LLM system, and the results demonstrate that our approach offers enhanced security, reduces the service latency for benign users, and decreases system resource consumption compared to state-of-the-art algorithms.
\end{abstract}

\begin{IEEEkeywords}
Prompt attack, edge-cloud, LLM, resource optimization, Bayesian game
\end{IEEEkeywords}

\section{Introduction}
Large language models (LLMs) are widely used in fields such as instruction understanding \cite{ouyang2022training}, information retrieval \cite{liu2024information}, software engineering \cite{nam2024using}, translation \cite{zhang2023prompting}, education \cite{kasneci2023chatgpt} and summarization \cite{koh2022empirical}. Companies are competing to launch their LLMs such as ChatGPT, Claude, and Qwen, which have significantly enhanced the efficiency of human activities.
Prompt engineering (PE) involves carefully crafting input prompts to ensure that LLMs generate outputs that meet user requirements, thus improving response quality \cite{sorensen2022information}. 
However, PE can also be exploited by attackers to design malicious prompts, resulting in prompt attacks that manipulate LLMs into producing unsafe content, including the ``grandma exploit" \cite{zhou2024quantifying} and jailbreak \cite{wei2024jailbroken}. These prompt attacks can have significant societal impacts, leading to the spread of misinformation \cite{gehman2020realtoxicityprompts}, discrimination \cite{kolla2024llm} and violations of social norms and morals \cite{NEURIPS2023_a2cf225b}. Additionally, they can enable unauthorized access to private data.
For example, Microsoft’s LLM-integrated Bing Chat and GPT4 can be compromised by prompt injection attacks that expose sensitive information \cite{greshake2023not}. Furthermore, LLMs require substantial computational resources, and processing malicious prompts increases system resource consumption and average latency for benign users \cite{liu2023optimizing}.

The safety of LLM outputs has attracted significant attention. 
To mitigate the risks associated with prompt attacks, model developers have implemented safety mechanisms through safety fine-tuning, utilizing the advanced technique of Reinforcement Learning from Human Feedback (RLHF).
Particularly, Dai \textit{et al.} proposed a Safe RLHF method to improve the security of LLM outputs \cite{dai2023safe}.
This method involves training a reward model (helpfulness) using preference data and a cost model (harmlessness) with safety data, followed by fine-tuning through the proximal policy optimization reinforcement learning method. However, Wei~\textit{et al.} point out that safety training will fail because the two objectives of helpfulness and harmlessness are sometimes in conflict \cite{wei2024jailbroken}. 

 


Compared to merely employing RLHF for safety tuning, 
proactive detection before processing can significantly enhance LLM system safety and improve Quality of Service (QoS) by reducing malicious resource consumption. While LLMs can be directly utilized to detect harmful prompts \cite{armstrong2022using}, the conflict between helpfulness and harmlessness in safety fine-tuning methods often leads to performance degradation. In addition, models with large numbers of parameters can be resource-intensive, negatively affecting system throughput.
However, detecting all prompts before inputting them into the LLM will increase resource consumption and service latency. To solve this problem, we constructed an Edge-Cloud LLM (EC-LLM) architecture that deploys the detection model at the edge, leveraging edge computing to reduce network throughput and latency \cite{wang2023edge}. 

The EC-LLM system faces many challenges. The first is how to improve the accuracy of detection objects. Existing anomaly detection models often face the problem of false positives \cite{bagheri2024ace}, where benign prompts are mistakenly classified as malicious, negatively impacting QoS. Thus, we innovatively deployed a \textit{vector database} (VDB) at each edge server (ES) containing datasets of normal and malicious prompts. By evaluating the similarity between incoming prompts and the data in the VDB, as well as the results from a Bert-based detector, we aim to improve detection reliability.


The second challenge is how to design a strategy to decide whether to detect incoming prompts to optimize the trade-off between system security and performance.
Limited resources in ES can lead to increased service latency and resource consumption due to excessive use of detectors \cite{liang2024latency, liang2023grouping}. Accurately identifying malicious prompts can minimize service latency for benign users and decrease resource consumption for LLMs. Current game-based strategies cannot effectively formulate strategies based on dynamic environments.
To mitigate this issue, we employ a \textit{multiple-stage} Bayesian Game model to simulate the interactions between users and system defenders. A sequential marginal analysis method is used to predict the number of malicious prompts under conditions of incomplete information.
By solving the Bayesian equilibrium solutions at each stage, the allocation of detection resources in the EC-LLM system is optimized. 
This strategic approach ensures that detection resources are judiciously allocated, balancing the trade-off between security and performance.




In this work, we are the first to jointly optimize system resources and service latency in EC-LLMs under prompt attacks. We develop a novel EC-LLM architecture that integrates a Bert-based prompt detector and a VDB to enhance prompt security. We then formulate the joint prompt detection, latency, and resource optimization (PDLRO) problem as a multi-stage incomplete information Bayesian game model. To solve the Bayesian equilibrium at each stage, we proposed a belief update algorithm and a malicious number prediction algorithm. Our evaluations on a real-word testbed demonstrate that our proposed algorithm outperforms all baseline algorithms. 
The main contributions of this work can be summarized as follows:
\begin{itemize}
    \item To the best of our knowledge, this is the first attempt to optimize system resources and service latency in EC-LLMs under prompt attacks.
    We formulate the PDLRO problem as a multi-stage incomplete information Bayesian game model and identify the Bayesian equilibrium solution to address it.
    \item A sequential marginal analysis method is applied to predict the number of malicious prompts for system defenders. Additionally, to improve identification accuracy at each stage, a belief updating algorithm based on history data and VDB is proposed. This approach effectively reduces system service latency and resource consumption.
    \item A practical EC-LLM architecture system with a detector and a VDB is constructed to evaluate our method. We conducted experiments on this real-world system testbed, and experimental results demonstrate the efficacy of our strategy. Compared to the baseline strategy, the average latency per token is reduced by approximately 16.24\%, the average GPU memory consumption is reduced by nearly 17.7\%, and the average GPU utilization is reduced by about 17.87\%, respectively.
\end{itemize}

The rest of the paper is organized as follows. 
Section \uppercase \expandafter {\romannumeral 2} describes the system model and formulation. The multiple-stage Bayesian game model is presented in Section \uppercase \expandafter {\romannumeral 3}. Section \uppercase \expandafter {\romannumeral 4} gives the details of our resource optimization strategy. The experiment results are shown in Section \uppercase \expandafter {\romannumeral 5}. The related works are introduced in Section \uppercase \expandafter {\romannumeral 6}. We conclude \uppercase \expandafter {\romannumeral 7}.

\section{System Model and Problem Formulation}

\subsection{System Model}
In our EC-LLM scenario, the task set is denoted by $\mathcal{K}= \{1, \dots, K\}$, generated by a user set $\mathcal{U}=\{1,\cdots, U\}$. Tasks include benign tasks $\mathcal{K}^{\text{b}}$ from benign users $\mathcal{U}^{\text{b}}$ and malicious users $\mathcal{U}^{\text{m}}$, and malicious tasks $\mathcal{K}^{\text{m}}$ from malicious users $\mathcal{U}^{\text{m}}$. The sets satisfy  $\mathcal{K}^{\text{b}} \cup\mathcal{K}^m= \mathcal{K}$ and $\mathcal{U}^{\text{b}} \cup\mathcal{U}^{\text{m}}= \mathcal{U}$. Each task $k\in\mathcal{K}$ is characterized by four attributes: the input prompt $x_k$, the generating user $u\in\mathcal{U}$, the required floating-point of operations per second (FLOPS) $f_k$, and the number of prompt tokens $o_k$, which is obtained by using a tokenizer to convert the prompt into tokens within the LLM. The LLM is tasked with generating an appropriate response for each prompt.

In our system, there is a set of $M$ ESs, denoted as $\mathcal{M}^e=\{1, \dots, M \}$. 
The GPU FLOPS capability available on these ESs is defined as $ \mathcal{G}^e=\{G_1, \dots, G_M \}$. Additionally, the GPU bandwidth for each server is specified as $\mathcal{B}^e =\{B^e_1, \dots, B^e_M\}$. Each ES has an owner, denoted as $\mathcal{L}=\{L_1, \dots, L_{M} \} $, and each is considered a defender. 
Each ES hosts a VDB, denoted as $\mathcal{V} = \{V_1, \dots, V_M \}$, which stores data in the format $\{v_q, s_q \}_{q=1}^{Q}$. Here, $v_q$ is a vectorized prompt derived from security datasets. The safety label $s_q \in \{0,1\}$ indicates whether $v_q$ is benign. Specifically, $s_q=1$ denotes that $v_q$ is malicious, while $s_q=0$ indicates a benign prompt. The architecture of the EC-LLM system is depicted in Fig.~\ref{ari}.

\begin{figure}
\centering

\includegraphics[width=0.48\textwidth]{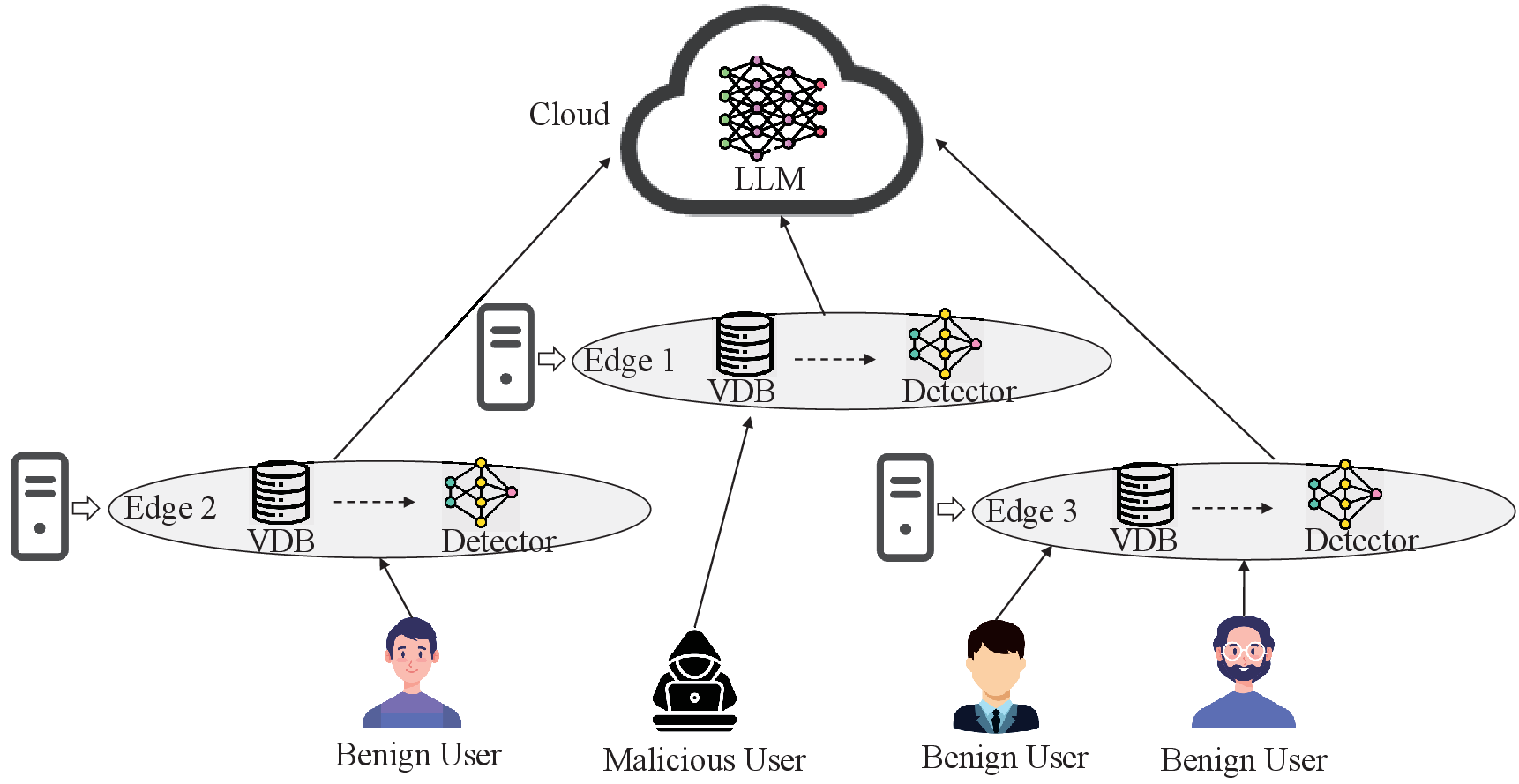}
	\caption{An illustration of EC-LLM architecture.}
	\label{ari}
 \vspace{-1em}
\end{figure}

\subsection{Security Detector}
 Malicious prompts pose a risk of generating unsafe outputs, leading to resource wastage when processed by LLMs. To differentiate between benign and potentially malicious prompts, a Bert-based detector is deployed on each ES to ascertain the safety of prompt $x_k$. If a malicious prompt is detected by the detector, it will not be sent to the cloud. Instead, the user will receive a direct response stating: ``This prompt is unsafe."

\subsection{Cost Model}

The detector's response is categorized as either benign or malicious. 
Consequently, the detection time and resource usage for the detector can be estimated. The number of parameters in the detector is denoted by $A$. For the transformer-based model, according to OpenAI \cite{kaplan2020scaling}, the floating-point of operations (FLOPs) per token is:
\begin{align}
\label{eq2}
C^{f} = 2A + 2n_{layer}n_{ctx}d_{attn},
\end{align}
where $n_{layer}$ is the number of layers, $n_{ctx}$ is the maximum number of input tokens and $d_{attn}$ is the dimension of the attention output. In the inference process, for the detector, the parameters need to be loaded into GPU memory. The GPU bandwidth is denoted as $B^g$ (TB/s). For a batch size, the memory latency $t^{mem}$ (ms) is calculated as:
\begin{align}
\label{eq221}
t^{mem} = 2A/B_g,
\end{align}
where $2A$ indicates each parameter is 2 bytes. The FLOPs of the deployed ES detector are denoted as $C^g$. The computational latency for the detector with one batch size is calculated as follows:
\begin{align}
\label{eq222}
t^{com} = C^{f}/C^g.
\end{align}
When the batch size is small, the memory latency exceeds the computational latency, creating a bottleneck in detector inference. Hence, the detection time for prompt $x_k$ is
\begin{align}
\label{eq223}
d_k^d = o_k \cdot \max\{t^{mem},t^{com}\}.
\end{align}
The detection process may increase the latency for benign users. A binary variable $z_k\in\{0,1\}$ is used to indicate whether prompt $x_k$ is detected. If $z_k=1$, this prompt is detected; Otherwise $z_k=0$. The $d_k$ represents the latency for prompt $x_k$ after the detection step, and it can be calculated as follows: 
\begin{align}
\label{eq31}
d_k = z_k\cdot d_{k}^d + d_{k}^{L}, \forall k \in\mathcal{K}, 
\end{align}
where $d_{k}^{L}$ is the latency of LLM for processing task $k$. 
For LLM, the number of output tokens is unknown, so the processing time cannot be predicted accurately. We define a binary variable $\eta_k$ to express the detection result of the detector. If $\eta_k=0$, the detection result is benign. $\eta_k=1$ means that prompt $x_k$ is malicious. If the prompt $x_k$ is not selected for detection, $\eta_k = 0$. The resources consumed after detection for task $k$ can be calculated as:
\begin{align}
\label{eq32}
C^D_k = z_k( o_k C^f - \eta_k\cdot f_k)+f_k.
\end{align}



\subsection{Formulation}
We divide time into $T$ time slots. For time slot $t \in \mathcal{T}$, the task set $\mathcal{K}_t$ is defined as $\mathcal{K}_t = \{k_{t,1},\cdots,k_{t,{K_1}} \}$, where $K_1$ is the number of prompts at time slot $t$ and $\mathcal{K} = \bigcup_{t=1}^T \mathcal{K}_t$. 
Each task $k_{t,\lambda}$ has four attributes: the input prompt $x_{t,\lambda}$, the source user $u'\in \mathcal{U}_t$, the required FLOPS $f_{{t,\lambda}}$, and the number of tokens $o_{{t,\lambda}}$. 
The objective of our PDLRO problem is to reduce the total latency for benign users and reduce the resources consumed by malicious users and detection, which can be formulated as follows:
\begin{align}
     \min \sum_{t\in\mathcal{T}} (\ & \alpha_1 \sum_{k_{t,\lambda}\in \mathcal{K}_b} d_{k_{t,\lambda}} + \alpha_2 \sum_{k_{t,\lambda}\in \mathcal{K}} C_{k_{t,\lambda}}^D) \\
    \text{s.t.} \ \ \
    & z_{k_{t,\lambda}} \in \{0,1\}, \forall k_{t,\lambda} \in \mathcal{K}_t. \\
    & \eta_{k_{t,\lambda}} \in \{0,1\}, \forall k_{t,\lambda} \in \mathcal{K}_t. \\
    & \alpha_1, \alpha_2 > 0,
\end{align}
where $\alpha_1$ and $\alpha_2$ are two weights used to balance latency and resource consumption. 
The problem is inherently challenging as it is an online problem where the detected result $\eta_{k_{t,\lambda}}$ is unknown before making the detection decision $z_{k_{t,\lambda}}$. 
\begin{theorem}
This problem is NP-hard. 
\end{theorem}
\begin{proof}
The Knapsack Problem involves selecting items with specific weights and values to maximize the total value within a weight limit, and it is NP-hard. We can map it to our problem. 

In the Knapsack Problem, each item has a weight and a value. In our problem, each task prompt $x_k$ has required FLOPS (analogous to weight) and $\alpha d_{k_{t,\lambda}} + \beta C_{k_{t,\lambda}}^D$ analogous to value. The Knapsack Problem has a maximum weight capacity. In our problem, the ESs have a limited amount of GPU FLOPS and bandwidth. The decision variable of the Knapsack Problem is whether the item is included in the knapsack. In our problem, the decision variable is whether prompt $x_k$ is detected by the ES. 
The objective of the Knapsack Problem is to maximize the total value of selected items. In our problem, the objective is to minimize the total latency for benign users and the resources consumed by malicious users and detection. 
By constructing an instance of our problem that mimics the structure of the Knapsack Problem, we can show that solving our problem requires solving the Knapsack Problem as a subproblem, thus establishing NP-hardness.
\end{proof}



\section{Bayesian Game Model}
\subsection{Game Model Formulation}
The PDLRO problem involves three parties: benign users, malicious users, and defenders. The interactions among these parties are complex, and the defenders' decisions regarding whether to conduct detection must be based on information. These characteristics make it highly appropriate to model the situation using a $T$-stage Bayesian game with incomplete information, corresponding to $T$ time slots.


Our game model is based on the following assumptions:
\begin{enumerate}
\item All players are rational.
\item Each player is unaware of the strategies, actions, or utility functions of other players.
\item The total number of players remains constant over time.
\item Malicious players randomly send either benign or malicious prompts. If a malicious player sends a benign prompt, they are considered a benign user at this stage.
\end{enumerate}

The basic $T$-stage game model is denoted as $\mathcal{G}_t=(\mathcal{N},\Phi, \mathcal{A}_t, \mathcal{P}_t )$. The set of players is defined as $\mathcal{N} = \{ \mathcal{U}^{\text{b}}, \mathcal{U}^{\text{m}}, \mathcal{L} \}$. The set of types, $\Phi$, includes all types for all players, specifically $\Phi=\{\text{b},\text{m}, \text{d} \}$, corresponding to benign users, malicious users, and defenders, respectively. The set of actions at time slot $t$, $\mathcal{A}_t$, is defined as $\mathcal{A}_t=\mathcal{A}_t^{\text{b}} \times \mathcal{A}^{\text{m}}_t \times \mathcal{A}^{\text{d}}_t$. The action set for benign users $\mathcal{A}_t^{\text{b}}= \mathcal{M}^e$, involves selecting an ES to minimize latency. The action for malicious users, $\mathcal{A}_t^{\text{m}}= \mathcal{M}^e$, involves selecting an ES to target. The defender's action set, $\mathcal{A}^{\text{d}}=\{2^{\mathcal{K}_t}\}$, is defined as the power set of $\mathcal{K}_t$, representing various combinations of prompts used for detection. The function $p_t(\phi)\in\mathcal{P}_t$ represents the prior probability that a player belongs to type $\phi\in\Phi$. 

The information structure at time slot $t$ is denoted by $\Theta_t= \{ (h_{t,n})_{n\in\mathcal{N}}, \tau(\cdot,\cdot) \}$. It includes $h_{t,n}$, which represents the private information of each player at time slot $t$, inaccessible to other players. $\tau(\cdot,\cdot)$ denotes the common prior belief shared among all participants. The information structure is designed to meet the following criteria:
\begin{align}
\label{tij}
\sum_{{n}\in \mathcal{N}} \tau(h_{t,n},\phi) = p_t(\phi), \forall \phi\in\Phi.
\end{align}

The basic $T$-stage game and the information structure form an incomplete information game. 
In this game model, a pure strategy is a definite action choice, while mixed strategies involve randomness to maximize the total reward. The strategy of player $n$ is defined as $\pi_{t,n}:h_{t,n}\to \mathcal{A}_t $, mapping the information set to the action set. $\mu(\phi|h_{t,n})$ is the belief of player $n$, which represents the estimate of type $\phi$ after receiving information $h_{t,n}$.

\begin{figure}
	\centering
	\includegraphics[width=0.5\textwidth]{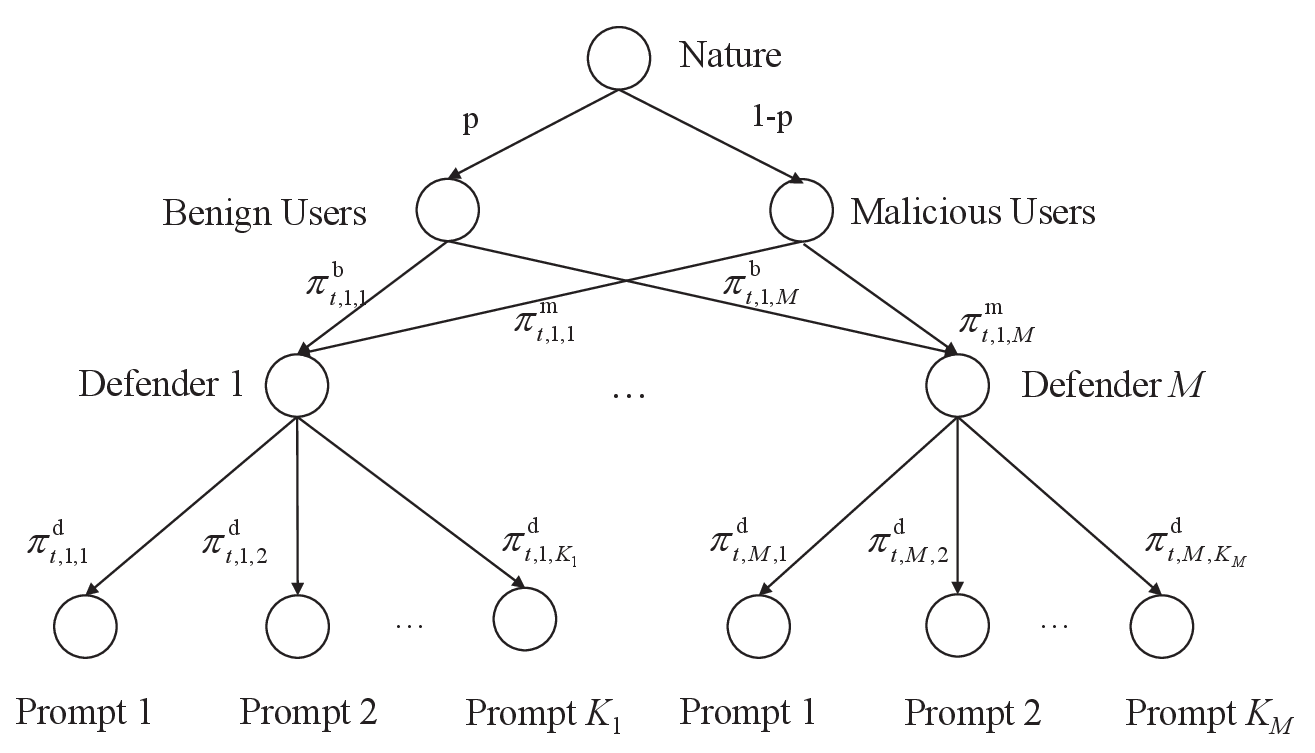}
	\caption{The game tree of our game model.}
	\label{tree}
\end{figure}




For benign player $n \in \mathcal{N}$, the strategy is denoted as $\pi^{\text{b}}_{t,n,m}$, which represents the probability that player $n$ selects edge $m$ at time slot $t$. The strategies of other players are defined as $\pi_{t,-n,m}^{\text{b}}$. The utility function is denoted as follows:
\begin{align}
\label{eq13}
 U_{t,n}^{\text{b}}(\pi^{\text{b}}_{t,n,m}, \pi_{t,-n,m}^{\text{b}}) = \sum_{m \in\mathcal{M}}\pi^{\text{b}}_{t,n,m} d_{t,n,m}, t \in \mathcal{T}, 
\end{align}
where $d_{t,n,m}$ denotes the latency experienced by player $n$ when a task is processed on ES $m$. Additional latency arises if benign prompts are incorrectly identified during detection.
The utility function for malicious user $n\in\mathcal{N}$ is:
\begin{align}
\label{eq14}
 U_{t,n}^{\text{m}}(\pi^{\text{m}}_{t,n,m}, \pi_{t,-n,m}^{\text{m}}) & = \sum_{m \in\mathcal{M}} [ \pi^{\text{m}}_{t,n,m} ((1-p^D)f_{k_{t,\lambda}} \\ \notag
& + o_{k_{t,\lambda}}C^{f}) + (1- \pi^{\text{m}}_{t,n,m})f_{k_{t,\lambda}}],
\end{align}
where $p^D$ is the F1 score of detector.
The utility function for defenders is: 
\begin{align}
\label{eq15}
 U_{t,n}^{\text{d}}(\pi^{\text{d}}_{t,n,\lambda}, \pi^{\text{d}}_{t,-n,\lambda}) = & \alpha_3 \sum_{k_{t,\lambda}\in \mathcal{K}_t} \pi^{\text{d}}_{t,n, \lambda} \cdot d_{t,n,\lambda}
+ \notag \\ \alpha_4 \sum_{k_{t,\lambda}\in \mathcal{K}_t} (\pi^{\text{d}}_{t,n,\lambda} (o_{k_{t,\lambda}}C^f 
 + & \mu(\phi| h_{t,n,\lambda}) f_{k_{t,\lambda}})+f_{k_{t,\lambda}}),
\end{align}
where $\alpha_3$ and $\alpha_4$ are positive weights used to balance latency and resource consumption. 
Therefore, the problem can be reformulated as:
\begin{align}
     \min \ & U_{t,n}^{\text{d}}(\pi^{\text{d}}_{t,n,\lambda}, \pi^{\text{d}}_{t,-n,\lambda}) \\
    \text{s.t.} \ \ \
    & \pi^{\text{d}}_{t,n,\lambda} \in [0,1], \forall n\in\mathcal{L}, k_{t,\lambda} \in \mathcal{K}_t. \\
    & \sum_{k_{t,\lambda} \in \mathcal{K}_t.} \pi^{\text{d}}_{t,n,\lambda} = K'_{t,n} , \forall n\in \mathcal{L}. \\ 
    & \mu(\phi|h_{t,n,\lambda}) \in [0,1], \forall k_{t,\lambda} \in \mathcal{K}_t, n\in \mathcal{L},
\end{align}
where $K'_{t,n}$ is the estimated number of malicious tasks that need to be detected for defender $n$. In this formulation, the strategy $\pi^{\text{d}}_{t,-n,\lambda}$ substitutes for the variable $z_{k_{t, \lambda}}$, while the belief $\mu(\phi|h_{t,n,\lambda})$ takes the place of the unknown $\eta_{k_{t,\lambda}}$, thereby converting the original problem into a linear format when $K'_{t,n}$ is estimated. This transformation allows for more flexible strategy formulation, enabling the capture of uncertainty in player interactions and decision-making processes. In the context of rational players, it effectively captures the underlying dynamics of the game, particularly when players formulate their actions based on expected outcomes.

\subsection{Game Model Analysis}
The PDLRO problem is a dynamic game with incomplete information. Harsanyi transformation is a standard method to convert an incomplete information game to an imperfect information game \cite{harsanyi1967games}. Initially, a prior move by Nature is introduced to determine the type of each player. Each player knows their type but unknown other players' type. Subsequently, benign and malicious players select an ES, after which defenders receive tasks and decide whether to conduct detection. The gaming process is illustrated by the game tree in Fig.~\ref{tree}. Since we use Bayes' rule to update beliefs, the game model is a multi-stage Bayesian Dynamic Game.

For a benign or malicious player $n$ ($\phi\in\{\text{b},\text{m}\}$ and $m\in\mathcal{M}^e$), the vector $ \pi_{t,{n}}^{\phi} = \{\pi_{t,{n},1}^{\phi}, \cdots, \pi_{t,n,m}^{\phi}, \cdots, \pi_{t,{n},M}^{\phi} \}$ has the same dimension as the number of ESs $M$. For the defender, the dimension of $\pi_{t,{n}}^{\text{d}}$ corresponds to the number of tasks $K_1$. 
A strategy combination $(\pi^{\text{b},*}_{t,n}, \pi^{\text{m},*}_{t,n}, \pi^{\text{d},*}_{t,n})$ is a Bayesian Nash equilibrium (BNE) of $\mathcal{G}_t$ if, at each time slot $t$, for each player $n\in\mathcal{N}$ and $\phi\in\Phi$, the following conditions hold:
\begin{align}
\label{eq3}
 U_{t,n}^{\phi}(\pi^{\phi,*}_{t,n}, \pi_{t,-n}^{\phi,*}) \geq U_{t,n}^{\phi}(\pi_{t,n}^{\phi}, \pi_{t,-n}^{\phi,*}), \forall t\in\mathcal{T}.
\end{align}
where $\pi_{t,-n}^{\phi,*}$ is the optimal strategy of opponents and $\pi_{t,n}^{\phi,*}$ is the optimal strategy for player $n$. 

\begin{theorem}
In our game model, the pure and mixed strategy equilibrium exists.
\end{theorem}
\begin{proof}
In our game model, the type space $\Phi$ includes the benign user, malicious user, and the defender, and is finite. The pure strategy space, corresponding to the action set, is also finite. Consequently, both type and strategy spaces are compact and convex. 
For benign players, $d_{t,n,m}$ is constant. According to Eqn.~(\ref{eq13}), $U_{t,n}^{\text{b}}$ is linear. Eqn.~(\ref{eq14}) can be transformed into:
\begin{align}
\label{eq142}
  U_{t,n}^{\text{m}}(\pi^{\text{m}}_{t,n,m}, \pi_{t,-n,m}^{\text{m}}) & = \sum_{m \in\mathcal{M}} [ \pi^{\text{m}}_{t,n,m} (p^D f_{k_{t,\lambda}} \\ \notag
  & + o_{k_{t,\lambda}}C^{f}) + f_{k_{t,\lambda}}],
\end{align}
where $f_{k_{t,\lambda}}$, $o_{k_{t,\lambda}}$ and $C^f$ are constants. Thus, the utility function is linear. For defender $n$, both $\mu(\phi | h_{t,n,\lambda})$ and $f_{k_{t,\lambda}}$ are constants. The utility function can be transformed into:
\begin{align}
\label{eq152}
 U_{t,n}^{\text{d}}(\pi^{\text{d}}_{t,n,\lambda}, \pi^{\text{d}}_{t,-n,\lambda}) = &  \sum_{k_{t,\lambda}\in \mathcal{K}_t} \pi^{\text{d}}_{t,n, \lambda} \cdot (\alpha_3 d_{t,n,\lambda}
+ \notag \\ \alpha_4 ( o_{k_{t,\lambda}}C^f 
 + & \mu(\phi| h_{t,n,\lambda}) f_{k_{t,\lambda}})+f_{k_{t,\lambda}})).
\end{align}

This function is also linear. 
For every type of player, the utility function is linear, so the utility function is quasi-concave. According to \cite{fudenberg1991game}, our game model satisfies these conditions that both type and strategy space are compact and convex, and the utility function is quasi-concave, and then the pure strategy equilibrium exists. 

According to \cite{fudenberg1991game}, in any finite non-cooperative game, there must be at least one Nash equilibrium in the sense of mixed strategies. Our game is a non-cooperative game, so a mixed strategy exists.
\end{proof}


\begin{theorem}
In our multi-stage game model, mixed strategies are at least as effective as pure strategies.
\end{theorem}
\begin{proof}
At time slot $t$, for player $n$, given the strategy of other players $ \pi_{t,-n_{t}}^{\phi}$, the strategy of player $n$ is calculated by:
\begin{align}
\label{eq23}
 \pi_{t,{n}}^{\phi} = \arg \min U_{t,n}^{\phi}(\pi_{t, n_{t}}^{\phi},\pi_{t,-n}^{\phi}), \phi \in \Phi. 
\end{align}

If $ \pi_{t,{n}}^{\phi}$ is a pure strategy and $\phi \in \{\text{b},\text{m} \}$, then $\pi_{t,{n},m}^{\phi}\in \{0,1\}$ and $\sum_{m\in\mathcal{M}^e} \pi_{t,{n},m} =1$. For a defender player $n$ and $k_{t,\lambda} \in \mathcal{K}_t$, $\pi_{t,{n},\lambda}^{\text{d}} \in \{0,1\}$. If $ \pi_{t,{n}}^{\phi}$ is a mixed strategy, then $\pi_{t,{n},m}$ and $\pi_{t,{n},\lambda}^{\text{d}} \in [0,1]$. Furthermore, given that the utility function is linear and quasi-concave, a pure strategy is a special case of a mixed strategy. Therefore, the mixed strategies are at least as effective as pure strategies.
\end{proof}

\section{Game Model-based Detection Resource Allocation}
In this section, we propose our Game Model-based Detection Resource Allocation (GMDRA) method.


\subsection{VDB}
VDB aids decision-making through belief updates. Defender $n$ reviews the most recent detection results and vector similarity metrics. The vector similarity between prompt $k_t$ and benign data is denoted as $p^s_{t,n}$, whereas the vector similarity with malicious data is represented as $p^{us}_{t,n}$. The defender’s belief indicates the safety probability for each prompt.

\subsection{Belief Updating}
Benign users, malicious users, and defenders all influence one another. Malicious users can mislead detection, increasing benign users’ latency. At each time slot $t$, all parties act on updated beliefs from prior information. For a benign player $n$, minimizing latency is crucial, relying on the previous response $y_{t-1,n}$, and the response latency $d_{t-1,n}$. Malicious users, with predicted number of tokens $o_{t-1,n}$, seek to consume more resources and increase latency. Each user's belief reflects their estimated detection probability.
The belief update algorithm is shown in Algorithm~\ref{alg1}. 

\begin{algorithm}[!t]
    \caption{Belief update algorithm} 
    \label{alg1} 
    \begin{flushleft}
    \textbf{Input}: Belief $\mu(\phi|h_{t-1,n})$, latency $d_{t-1,n}$, detection threshold $d_{\epsilon}$, token number $o_{t-1,n} $, predicted token number $o^{pre}_n$, safe similarity $p^s_{t,n}$, unsafe similarity $p^{us}_{t,n}$.
    
    \textbf{Output}: Belief $ \{\mu(\phi|h_{t,n})\}_{n\in\mathcal{N}}$. 
    \end{flushleft}
    \begin{algorithmic}[1] 
    \FOR {$n \in \mathcal{N} $ }
        \IF{$n\in\mathcal{U}^{\text{b}}$}
            \IF{ $y_{t-1,n}$ is benign}
                \STATE $p^{ld} = \gamma_1$, $p^{nld} = \gamma_2$
            \ELSE
                \STATE $p^{ld} = 1- \gamma_1$, $p^{nld} = 1 - \gamma_2$
            \ENDIF
            \STATE \textit{diff} $= d_{t-1,n}/o_{t-1,n} - d_{\epsilon}$ 
            \IF{\textit{diff} $> 0$}
                \STATE $p^{ld} = p^{ld} (1+100\cdot \textit{diff}) $
                \STATE $p^{nld} = p^{nld} (1/(1+100\cdot \textit{diff}))$
                \STATE $p^{ld} = min( p^{ld},1)$
            \ELSE
                \STATE $p^{ld} = p^{ld} (1 - 100\cdot abs(\textit{diff})) $
                \STATE $p^{nld} = p^{nld} (1/(1 - 100\cdot abs(\textit{diff})))$
                \STATE $p^{nld} = min( p^{nld},1)$
            \ENDIF
        \ENDIF
        \IF{$n\in\mathcal{U}^{\text{m}}$}
            \IF{ $y_{t-1,n}$ is malicious}
                \STATE $p^{ld} = \gamma_3$, $p^{nld} = \gamma_4$
            \ELSE
                \STATE $p^{ld} = 1- \gamma_3$, $p^{nld} = 1 - \gamma_4$
            \ENDIF
            \STATE $p^{ld} = p^{ld} (o^{pre}_{n}/o_{t-1,y_n} ) $
            \STATE $p^{nld} = p^{nld} (1 - o_{t-1,n}/o_{t-1,y_n} )$
            \STATE $p^{ld} = min( p^{ld},1)$, $p^{nld} = min( p^{nld},1)$
        \ENDIF
        \IF{$n\in\mathcal{U}^{\text{d}}$}
            \IF{ $y_{t-1,n}$ is malicious}
                \STATE $p^{ld} = \gamma_5$, $p^{nld} = \gamma_6$
            \ELSE
                \STATE $p^{ld} = 1- \gamma_5$, $p^{nld} = 1 - \gamma_6$
            \ENDIF
            \STATE $p^{ld} = p^{ld} p^s_{t,n} /(p^s_{t,n} + p^{us}_{t,n} ) $
            \STATE $p^{nld} = p^{nld} p^{us}_{t,n} /(p^s_{t,n} + p^{us}_{t,n} ) $
        \ENDIF
        \STATE $\mu(\phi|h_{t,n})= p^{ld} \mu(\phi|h_{t-1,n} ) / (p^{ld} \mu(\phi|h_{t-1,n}) + p^{nld} (1-\mu(\phi|h_{t-1,n}) )$
    \ENDFOR
    \end{algorithmic}
\end{algorithm}

Algorithm~\ref{alg1} updates beliefs based on the different types of players. Players assess their prompt's detection status and calculate the base likelihood values $p^{ld}$ and $p^{nld}$ based on the responses received (lines 3-6, 20-24 and 31-35). For benign user $n$, $d_\epsilon$ represents the ideal latency per token. The update to the likelihood function considers the difference between actual and ideal response latency (lines 10-11 and 14-15), capping values at 1 using the min function to ensure logical consistency. For malicious player $n$, the likelihood function is adjusted based on the actual versus expected number of response tokens (lines 25-26). Defender $n$ updates the likelihood using vector similarity measures (lines 35-36). Finally, beliefs are updated using the Bayesian formula in line 38. For the three different user types, the time complexity of conditional judgments, probability calculations, and belief updates are all $O(1)$, resulting in an overall time complexity of $O(N)$.



%

\subsection{Defender Strategy} 
After determining their beliefs, players act accordingly by choosing an ES to send their prompts. Both benign and malicious players contribute to a prompt queue, denoted as $\mathcal{K}_{t,m}$, from which the defender must develop a strategy to for detecting specific prompts.

However, defender $n$ lacks knowledge of the exact count of malicious prompts $K'_{t,n}$. To address this, we proposed a prediction algorithm based on sequential marginal analysis, which is an economic method that examines the marginal costs and benefits during continuous decision-making \cite{starrett1988foundations}. The prediction algorithm is detailed in Algorithm~\ref{alg2}.

\begin{algorithm}[!t]
    \caption{Malicious number predict algorithm} 
    \label{alg2} 
    \begin{flushleft}
    \textbf{Input}: The prompt set $\mathcal{K}_{t,m}$, and cost threshold $U^{cost}$.
    
    \textbf{Output}: Predicted number of malicious prompts $K'_{t,n}$.
    \end{flushleft}
    \begin{algorithmic}[1] 
    \STATE \textbf{Initialize}: $U = 0 $.
    \FOR{ $\lambda' \in \{0, \cdots, K_{t,m}\} $ } 
        \STATE Solve $opt(\lambda') = U^{\text{d}}_{t,n}$ under $K'_{t,n} = \lambda'$
        \IF{$\lambda' ==0$ }
            \STATE \textbf{continue} 
        \ELSE 
            \STATE $U = opt(\lambda') - opt(\lambda' -1)$
            \IF{$U > U^{cost}$ }
                \STATE \textbf{break}
            \ENDIF
        \ENDIF 
    \ENDFOR
    \STATE return $K'_{t,n} = \lambda' - 1$
    \end{algorithmic}
\end{algorithm}

In Algorithm~\ref{alg2}, we initiate a variable $U$ to record increments. In line 3, we compute the optimal value using the CLARABEL function from the Python cvxpy library with the predicted malicious prompts $\lambda'$. The increment is determined in line 7, and if $U$ exceeds a certain threshold, it indicates that further detections would entail excessive costs. Therefore, we return $\lambda'-1$. This method helps avoid the overallocation of detection resources. The complexity of CLARABEL depends on the number of variables, which in our game-theoretic framework includes three user types, two decision variables, and a finite set of constraints. As the total number of variables is a small constant that remains consistent across iterations, the overall complexity is effectively $O(1)$. Therefore, the time complexity of Algorithm \ref{alg2} is $O(K_{t,m})$.

Once the detection number $K'_{t,n}$ is established, the BNE strategy $\pi^{\text{d}}_{t,n,\lambda}$ is derived using a linear solver. Our GMDRA method then efficiently allocates detection resources to minimize system resource usage and latency.



\section{Evaluation}
In this section, comprehensive experiments are conducted to evaluate our methods. 

\subsection{Experimental Setup}
\textbf{System Setup}. Our system comprises four edge cloud clusters. Three devices serve as ESs, each equipped with an Intel(R) Core(TM) i9-10900K CPU featuring 10 physical cores, operating at a base frequency of 3.70 GHz, with a maximum turbo frequency of 5.3 GHz. These ESs run Rocky Linux 9.3 (Blue Onyx) with kernel version $5.14.0-362.8.1.el9\_3.x86\_64$, and are equipped with GeForce RTX 2070 SUPER GPUs. 
The fourth device functions as the cloud server, featuring an Intel Core i9-14900KF CPU with 24 cores (32 threads) and a base frequency of 3.20 GHz. It operates on Debian with kernel version 6.1.0-16-amd64, also featuring a GeForce RTX 4090 GPU. 

Deployed in the cloud is the Qwen1.5-7B-Chat\footnote{https://huggingface.co/Qwen/Qwen1.5-7B-Chat} LLM, accessible via an API developed from the api-for-open-llm project\footnote{https://github.com/xusenlinzy/api-for-open-llm}, encapsulated in a container image. Each ES hosts a Milvus VDB\footnote{https://github.com/milvus-io/milvus/tree/v2.3.3} using Docker Compose, at version v2.3.3, with public datasets for prompt injection and jailbreak uploaded for vector similarity matching. 
We employ the deberta-v3-base-prompt-injection-v2 detector\footnote{https://huggingface.co/protectai/deberta-v3-base-prompt-injection-v2}, a specialized model fine-tuned from Microsoft's deberta-v3-base \cite{he2021debertav3}, specifically designed to detect and classify prompt injection attacks.

\textbf{Datasets}. We use three datasets for testing and these data are not stored in the VDB.

\begin{enumerate}
    \item Classification\footnote{https://huggingface.co/datasets/jackhhao/jailbreak-classification}: This dataset includes 139 jailbreaks and 123 benign prompts. 
    \item Jailbreak-Benign-Malicious\footnote{https://huggingface.co/datasets/guychuk/benign-malicious-prompt-classification}: This dataset aims to help detect prompt injections and jailbreak intent.
    \item Malicious-Prompts\footnote{https://huggingface.co/datasets/ahsanayub/malicious-prompts}: 
This is a prompt injection and jailbreak dataset sampled from various sources.
\end{enumerate}



\subsection{The Effectiveness of the Architectures }
To evaluate the effectiveness of our architecture, we conducted an experiment comparing three setups: cloud-only, EC-LLM without detector, and EC-LLM with detector. In the cloud-only setup, users send prompts directly to the cloud. In the EC-LLM without detector architecture, prompts are sent to the edge and forwarded to the cloud without any filtering. In the EC-LLM with detector, a detector identifies and discards malicious prompts. 
Due to the large size of the Benign-Malicious and Malicious-Prompts datasets, we randomly selected 100 entries from each for testing, resulting in a total of 262 malicious and 200 benign prompts.

Throughput is measured as the number of output tokens per second from the LLM. Malicious prompts are regarded as positive, and the F1 score is used to measure the accuracy of the detector. The results are shown in TABLE.~\ref{tab1}. 

\begin{table}[h]
\centering
\caption{Throughput and F1 score of various architectures.}
\label{tab1}
\begin{tabular}{lccc}
\toprule
\textbf{}             &Cloud-only &\makecell[c]{EC-LLM\\ without detector} &\makecell[c]{EC-LLM\\ with detector} \\
\midrule
Throughput  & 43.089514801   & 43.475948882     & \textbf{46.140365851}  \\
F1 score & 0.804301075      & 0.804301075     & \textbf{0.927710843}     \\
\bottomrule
\end{tabular}
\end{table}

\begin{table*}[h]
\centering
\caption{Model Security Comparison}
\label{tab2}
\begin{tabular}{@{}lccccccccccc@{}}
\toprule
\makecell{Models} & & \multicolumn{3}{c}{Detector} & \multicolumn{3}{c}{LLM(GPT4)} & \multicolumn{3}{c}{LLM(Human)} \\
\cmidrule(lr){3-5} \cmidrule(lr){6-8} \cmidrule(l){9-11}
\makecell{Metric} & & Precision & Recall & F1 Score & Precision & Recall & F1 Score & Precision & Recall & F1 Score \\
\midrule
Jailbreak-Classification & & \textbf{0.9832} & \textbf{0.8417} & \textbf{0.9070} & 0.8571 & 0.6906 & 0.7649 & 0.9158 & 0.6259 & 0.7436 \\
Benign-Malicious & & 0.9762 & \textbf{0.8367} & \textbf{0.9011} & \textbf{1} & 0.44 & 0.6111 & 0.9565 & 0.44 & 0.6027 \\
Malicious-Prompts& & 0.9595 & \textbf{0.9726} & 0.9660 & \textbf{1} & 0.9315 & 0.9645 & \textbf{1} & 0.9452 & \textbf{0.9718} \\
\bottomrule
\end{tabular}
\vspace{-1em}
\end{table*}


From TABLE.~\ref{tab1}, it is clear that the throughput of the EC-LLM architecture exceeds that of the cloud-only setup, with the EC-LLM with detector further enhancing throughput. Moreover, the detector performs effectively on these datasets, achieving a higher F1 score compared to the direct use of the LLM without the detector. More details on the detection process in each dataset will be provided in the next section.

\subsection{The Effectiveness of detectors}
Given the dataset's imbalance, we utilize precision, recall, and F1 score to evaluate the performance of the models. For detector, true positive (TP) is defined as instances where both the detection result and prompt label indicate malicious. For the LLM, a malicious prompt that goes unanswered is considered TP, while receiving a normal response is classified as a false negative. A benign prompt receiving a normal response is a true negative, and if it is rejected, it is considered a false positive. We use both GPT-4 and manual evaluations to assess the LLM's responses.

The results for different models are presented in TABLE~\ref{tab2}.
From the table, it is evident that this LLM has some ability to resist prompt attacks, it is generally less effective than detectors. Deploying detectors at the edge effectively enhances the filtering of malicious prompts.

\subsection{The Influence of Different Detection Strategies}
To evaluate the effectiveness of our GMDRA with VDB strategy, we compare it against four strategies: None-detection, full detection, evolutionary genetic detection, and GMDRA without VDB. The evaluation involves a game configured with 20 players, divided into five stages, each lasting 75 seconds to account for task prompt processing times. Average latency per token is defined as the average time taken for each token (input + output) within a time slot. 
A fixed number of four malicious players randomly send benign or malicious prompts during each interval. The ES assignment for users is determined by the game. All tasks are completed within these time slots, and the results are illustrated in Fig.~\ref{fig4}.

Fig.~\ref{fig4a} presents the cumulative average latency per token for benign users across different strategies. The none-detection strategy incurs the highest latency, while the other strategies reduce latency by filtering out malicious prompts, which consume significant resources. The evolutionary genetic strategy further decreases latency compared to full detection. For the GMDRA without VDB approach, the absence of an integrated VDB forces the system to rely solely on historical data for evaluating current prompts, leading to an increase in latency as benign prompts from malicious players are incorrectly detected. In contrast, the GMDRA with VDB achieves the lowest latency widening over time. This approach reduces latency by approximately 16.24\% compared to the none-detection strategy.


In Fig.~\ref{fig4b} and Fig.~\ref{fig4c}, the trends for cumulative average GPU memory and GPU utilization show a consistent pattern. The full detection strategy identifies malicious prompts, resulting in lower average GPU memory consumption and reduced average GPU utilization compared to the none-detection approach. Although the evolutionary genetic strategy and the GMDRA without VDB method do lessen resource consumption compared to the none-detection strategy, they are not as effective as the full detection strategy. The GMDRA with VDB not only filters out malicious prompts but also reduces detection time. This results in a decrease in both cumulative average GPU memory and GPU utilization by shortening the task completion time within each time slot. Compared to the none-detection strategy, GMDRA with VDB reduces average GPU memory consumption by nearly 17.7\% and average GPU utilization by about 17.87\%.

\begin{figure*}[ht]  
\centering
\begin{minipage}[t]{0.31\textwidth}
    \centering
    \begin{subfigure}{1\linewidth}
	\includegraphics[width=\textwidth]{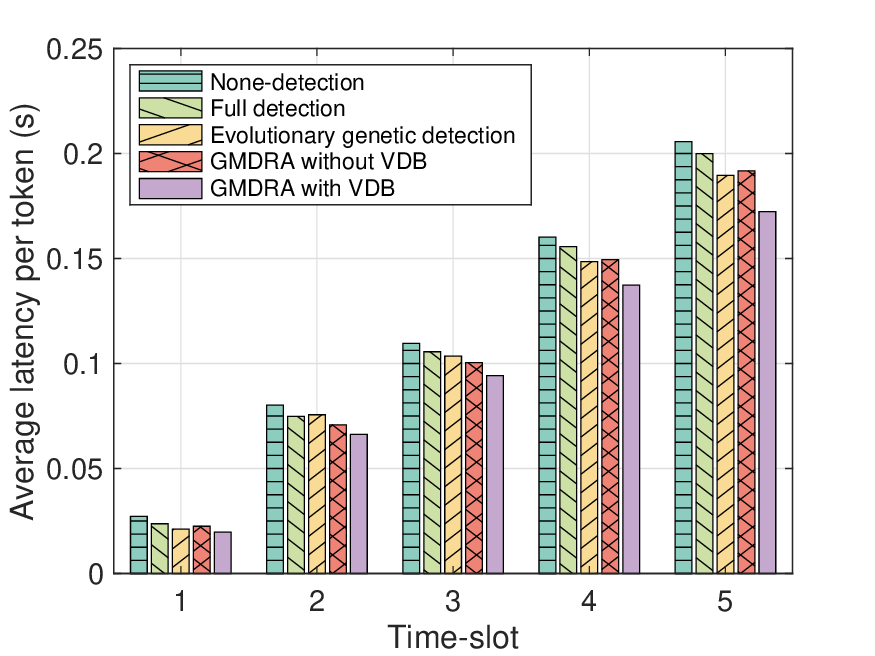}
	\caption{Average latency per token.}
	\label{fig4a}
    \end{subfigure}
\end{minipage}
\begin{minipage}[b]{0.31\textwidth}
    \centering
    \begin{subfigure}{1\linewidth}
        \includegraphics[width=\textwidth]{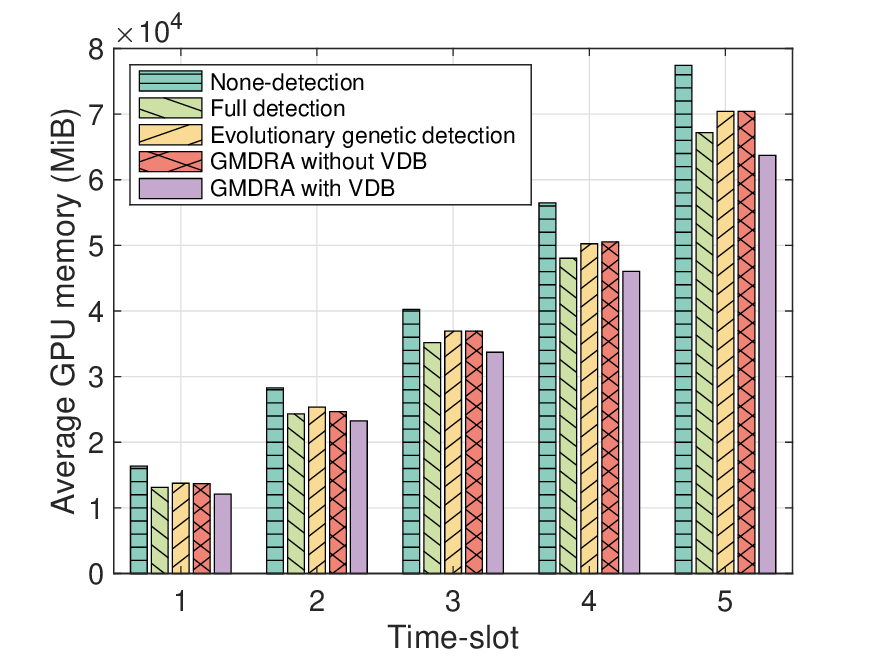}
        \caption{Average GPU memory.}
        \label{fig4b}
    \end{subfigure}
\end{minipage}
\begin{minipage}[b]{0.31\textwidth}
    \centering
    \begin{subfigure}{1\linewidth}
        \includegraphics[width=\textwidth]{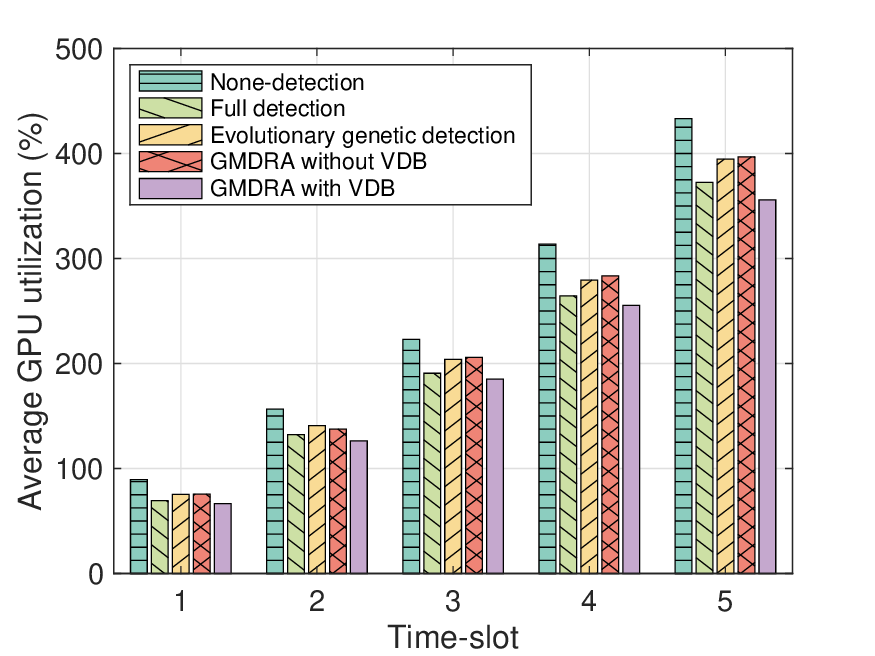}
        \caption{Average GPU utilization.}
        \label{fig4c}
    \end{subfigure}
\end{minipage}
\caption{The cumulative performance metrics for benign users under different strategies. }
\label{fig4}
\vspace{-0.8em}
\end{figure*}

\begin{figure*}[ht]
\centering
\begin{minipage}[t]{0.31\textwidth}
    \centering
    \begin{subfigure}{1\linewidth}
	\includegraphics[width=\textwidth]{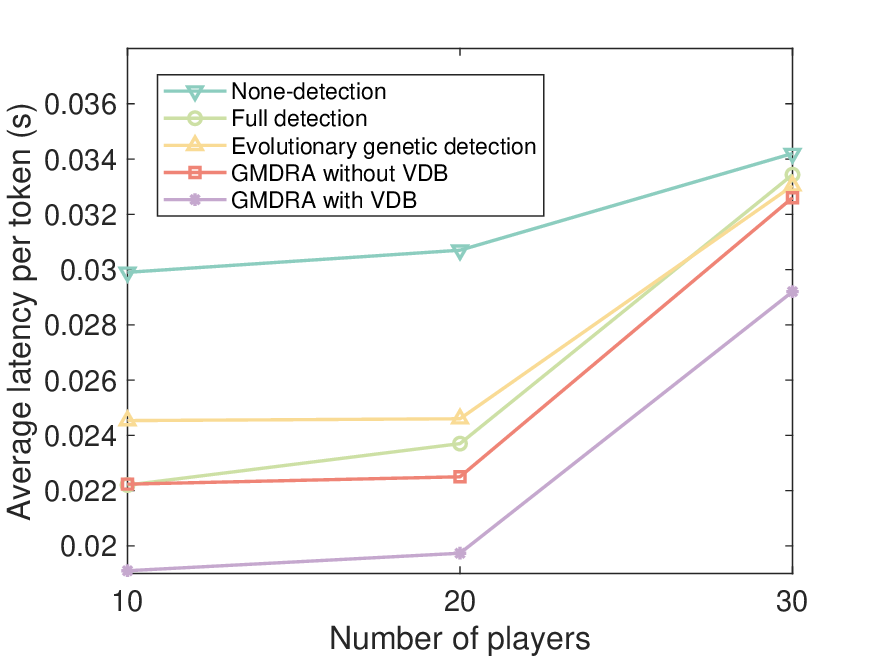}
	\caption{Average latency per token.}
	\label{fig5a}
    \end{subfigure}
\end{minipage}
\begin{minipage}[b]{0.31\textwidth}
    \centering
    \begin{subfigure}{1\linewidth}
	\includegraphics[width=\textwidth]{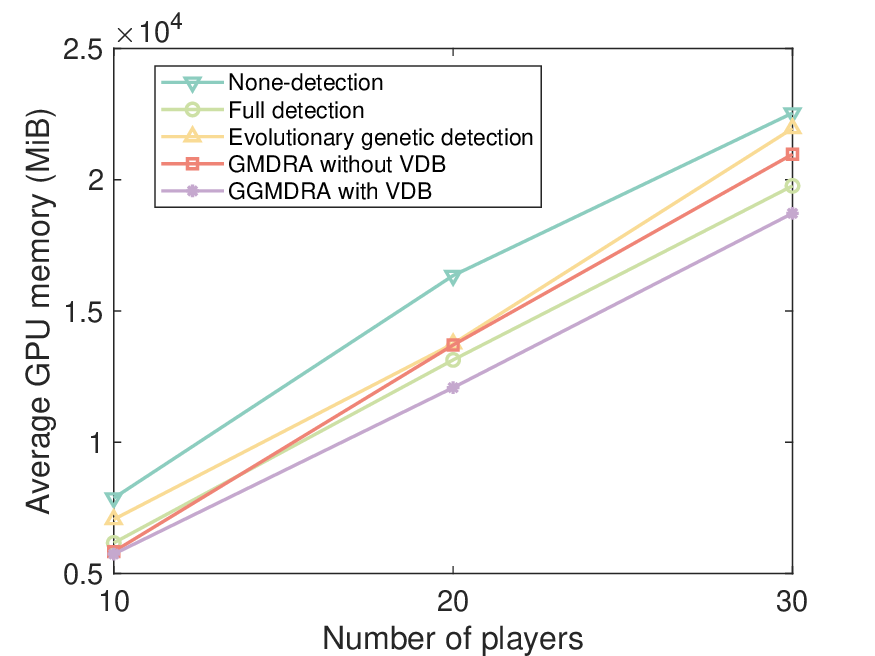}
	\caption{Average GPU memory.}
	\label{fig5b}
    \end{subfigure}
\end{minipage}
\begin{minipage}[b]{0.31\textwidth}
    \centering
    \begin{subfigure}{1\linewidth}
	\includegraphics[width=\textwidth]{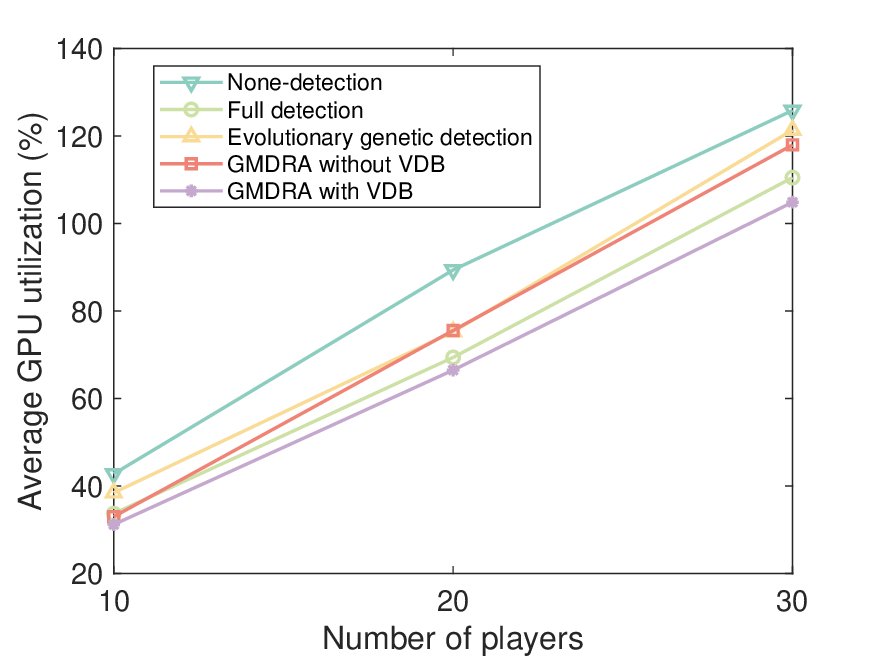}
	\caption{Average GPU utilization.}
	\label{fig5c}
     \end{subfigure}
\end{minipage}
\caption{The performance metrics for benign users under different numbers of players. }
\label{fig5}
\vspace{-0.8em}
\end{figure*}

\begin{figure*}[h]
\centering
\begin{minipage}[t]{0.31\textwidth}
    \centering
    \begin{subfigure}{1\linewidth}
	\includegraphics[width=\textwidth]{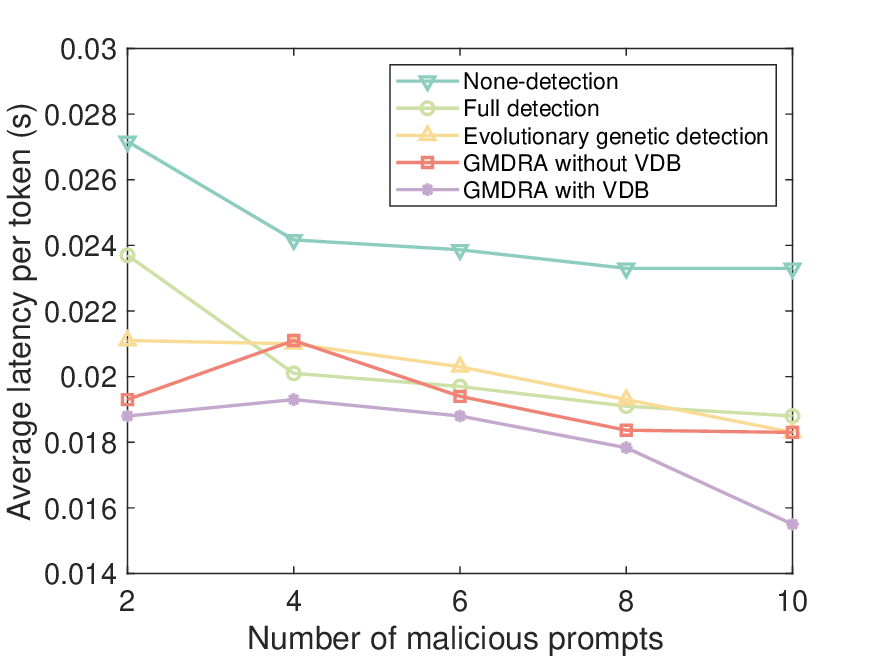}
	\caption{Average latency per token.}
	\label{fig6a}
    \end{subfigure}
\end{minipage}
\begin{minipage}[b]{0.31\textwidth}
    \centering
    \begin{subfigure}{1\linewidth}
	\includegraphics[width=\textwidth]{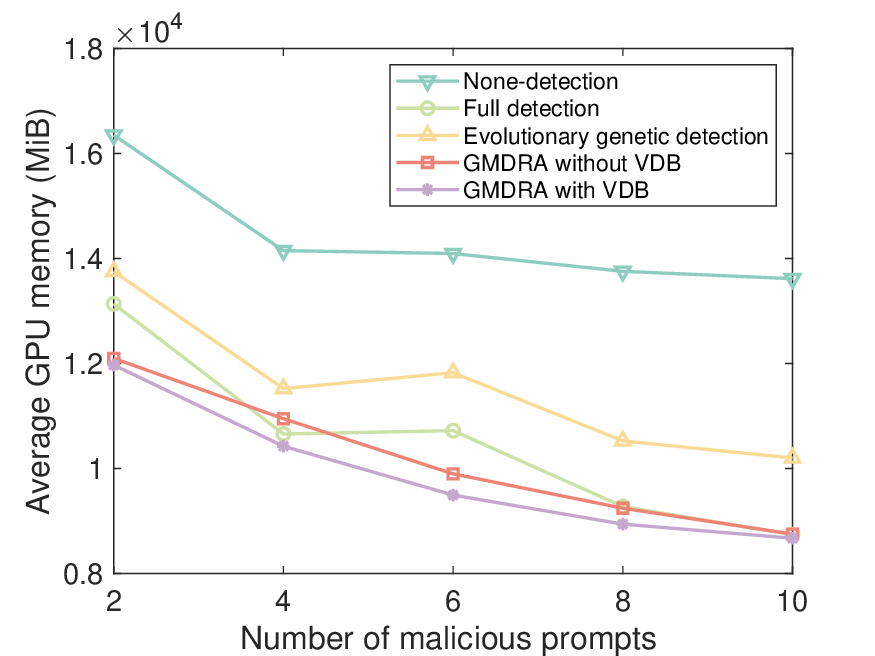}
	\caption{Average GPU memory.}
	\label{fig6b}
    \end{subfigure}
\end{minipage}
\begin{minipage}[b]{0.31\textwidth}
    \centering
    \begin{subfigure}{1\linewidth}
	\includegraphics[width=\textwidth]{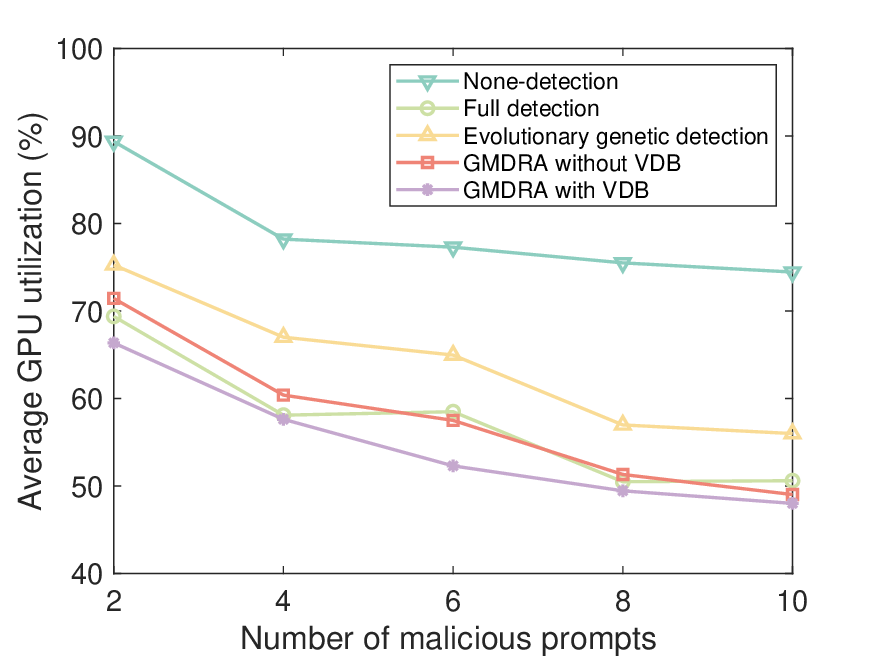}
	\caption{Average GPU utilization.}
	\label{fig6c}
     \end{subfigure}
\end{minipage}
\caption{The performance metrics for benign users under different numbers of malicious prompts. }
\label{fig6}
\vspace{-1em}
\end{figure*}



\subsection{The Influence of the Number of Players}
To explore the impact of the player count on various strategies, we design a comparative scheme with the number of malicious prompts fixed at 2 and the total number of players set to 10, 20, and 30. 

Fig.~\ref{fig5a} displays the average latency per token for benign users across different player counts. As the number of players increases, latency for all strategies increases. The average latency of none-detection is the largest, followed by the evolutionary genetic detection strategy. Among different numbers of players, our GMDRA with VDB strategy performs best.

Figs.~\ref{fig5b}-\ref{fig5c} show GPU memory usage and utilization for benign users at different numbers of players. As the number of players increases, the volume of prompts received by LLM increases, leading to a nearly linear increase in average GPU memory consumption and utilization. The non-detection strategy results in high GPU memory and usage, while other methods are more efficient in reducing resource consumption. Our GMDRA with VDB detection strategy demonstrates superior performance and stability across various player counts.

\subsection{The Influence of the Number of Malicious Prompts}
To assess the impact of different numbers of malicious prompts on our identification strategy, we conduct a comparative experiment with 20 players. Here, 10 players are malicious, and the total number of malicious prompts sent is varied at 2, 4, 6, 8, and 10.

In Fig.~\ref{fig6a}, dthe average latency per token remains relatively low despite an increase in malicious prompts. This is attributed to the decrease in benign prompts and the LLM's partial rejection of malicious submissions. When malicious players submit benign prompts, the GMDRA without VDB often misclassifies them, causing increased latency. In contrast, our GMDRA with VDB strategy consistently achieves the lowest latency for benign users.

In Figs.~\ref{fig6b}-\ref{fig6c}, average GPU memory usage and utilization decrease as the number of benign users declines. The none-detection strategy consumes significantly more GPU resources, while our GMDRA with VDB strategy shows the lowest GPU memory usage and utilization, proving effective across different scenarios with varying malicious prompts.




\section{Related Work}
\subsection{Prompt Attacks}
Prompt attacks have emerged as a significant threat to LLMs. Although security mechanisms have been developed for LLMs, they remain vulnerable to attacks such as jailbreak and prompt injection.
Wei \textit{et al.} highlighted the vulnerability of LLMs to prompt jailbreak attacks due to conflicting and competing training objectives \cite{ wei2024jailbroken}. They successfully executed jailbreak attacks on prominent models like GPT-4 and Claude v1.3, arguing that increasing model size and generalization capabilities without adequate defenses can intensify these vulnerabilities.

Xue \textit{et al.} \cite{xue2024trojllm} proposed TrojLLM, a black-box framework that generates toxic triggers to manipulate LLM outputs with specific prompts. Shi \textit{et al.} \cite{shi2022promptattack} contend that adversarial and backdoor attacks are feasible on prompt-based models, developing a malicious prompt template construction method through gradient search that diminishes LLM accuracy. Du \textit{et al.} demonstrated a poisoned prompt tuning method achieving a 99\% success rate by creating a direct shortcut between a trigger word and a target label, enabling attackers to manipulate model predictions with minimal input \cite{du2022ppt}. Greshake \textit{et al.} revealed that LLM-integrated applications blur the distinction between data and instructions, uncovering several new attack vectors \cite{greshake2023not}. They described how adversaries can exploit these applications remotely by injecting prompts into data likely to be accessed during inference.

\subsection{Defenses}
To address prompt attacks, various defense methods have been proposed, categorized into prevention-based and detection-based defenses \cite{liu2024formalizing}. Paraphrasing alters the text's appearance while retaining its meaning, while retokenization splits tokens in a prompt into smaller components. Both methods aim to neutralize malicious inputs from attackers but may unintentionally alter prompt content and affect response quality. 
Hines \textit{e al.} employed delimiting, datamarking, and encoding to improve LLMs' ability to differentiate between data and instructions, reducing the attack success rate from 50\% to 2\% with minimal impact on task effectiveness \cite{hines2024defending}. However, this approach could increase the prompt complexity and require additional computing resources.

For detection-based methods, adversarial suffixes are a feature of jailbreak attacks, which can be detected by learning adversarial suffixes \cite{ wang2024a}.
Caselli \textit{et al.} introduced HateBERT, an improved model for detecting abusive language that outperforms the general BERT model \cite{caselli2020hatebert}. While HateBERT effectively identifies objectionable content, it may fail to detect carefully crafted prompt attacks. 
However, these detection methods do not address issues related to resource consumption and latency optimization.


\subsection{Bayesian Game Model}
Harsanyi introduced Bayesian methods into game theory, formalizing the analysis of games with incomplete information \cite{harsanyi1967games}. In this framework, the concept of ``type" represents the private information of each participant, allowing for the application of Bayes' rule to update beliefs about other's types. 
Bayesian game models are widely used in network security. Yan \textit{et al.} proposed a Bayesian network game framework to analyze the dynamics of distributed denial-of-service attacks and defenses, elucidating the strategic interactions between attackers and defenders. 
Mabrouk \textit{et al.} designed an intrusion detection game based on a signaling game theory to secure data exchanges between ambulances and hospitals, using BNE to predict node types and enhance security \cite{mabrouk2023intrusion}. 
Ge \textit{et al.} \cite{ge2023gazeta} constructed  the GAZETA framework, employing a dynamic Bayesian game model for interdependent trust evaluation and authentication strategies, demonstrating significant improvements in network security.

These studies overlook the impact on benign users and do not simultaneously address system security and resource optimization, highlighting a gap in integrating comprehensive security measures with efficient resource management.

\section{Conclusion}
In this paper, we address the joint optimization of prompt security and system resource allocation in EC-LLM systems. We design an EC-LLM structure that deploys a detector and a VDB at each ES to enhance the security and efficiency of prompt processing. 
Subsequently, we formulate the PDLRO problem as a multi-stage Bayesian dynamic game model, developing solutions for predicting malicious prompts and updating beliefs through various stages, we also propose the GMDRA method to optimize detection resource allocation. Finally, tests on a real-world EC-LLM system demonstrate that our approach significantly enhances system security while reducing user latency and resource consumption.

\newpage
\bibliographystyle{IEEEtran}
\bibliography{main}

\begin{thebibliography}{10}
\providecommand{\url}[1]{#1}
\csname url@samestyle\endcsname
\providecommand{\newblock}{\relax}
\providecommand{\bibinfo}[2]{#2}
\providecommand{\BIBentrySTDinterwordspacing}{\spaceskip=0pt\relax}
\providecommand{\BIBentryALTinterwordstretchfactor}{4}
\providecommand{\BIBentryALTinterwordspacing}{\spaceskip=\fontdimen2\font plus
\BIBentryALTinterwordstretchfactor\fontdimen3\font minus \fontdimen4\font\relax}
\providecommand{\BIBforeignlanguage}[2]{{%
\expandafter\ifx\csname l@#1\endcsname\relax
\typeout{** WARNING: IEEEtran.bst: No hyphenation pattern has been}%
\typeout{** loaded for the language `#1'. Using the pattern for}%
\typeout{** the default language instead.}%
\else
\language=\csname l@#1\endcsname
\fi
#2}}
\providecommand{\BIBdecl}{\relax}
\BIBdecl

\bibitem{ouyang2022training}
L.~Ouyang, J.~Wu, X.~Jiang, D.~Almeida, C.~Wainwright, P.~Mishkin, C.~Zhang, S.~Agarwal, K.~Slama, A.~Ray \emph{et~al.}, ``Training language models to follow instructions with human feedback,'' \emph{Advances in neural information processing systems}, vol.~35, pp. 27\,730--27\,744, 2022.

\bibitem{liu2024information}
Z.~Liu, Y.~Zhou, Y.~Zhu, J.~Lian, C.~Li, Z.~Dou, D.~Lian, and J.-Y. Nie, ``Information retrieval meets large language models,'' in \emph{Companion Proceedings of the ACM on Web Conference 2024}, 2024, pp. 1586--1589.

\bibitem{nam2024using}
D.~Nam, A.~Macvean, V.~Hellendoorn, B.~Vasilescu, and B.~Myers, ``Using an llm to help with code understanding,'' in \emph{Proceedings of the IEEE/ACM 46th International Conference on Software Engineering}, 2024, pp. 1--13.

\bibitem{zhang2023prompting}
B.~Zhang, B.~Haddow, and A.~Birch, ``Prompting large language model for machine translation: A case study,'' in \emph{International Conference on Machine Learning}.\hskip 1em plus 0.5em minus 0.4em\relax PMLR, 2023, pp. 41\,092--41\,110.

\bibitem{kasneci2023chatgpt}
E.~Kasneci, K.~Se{\ss}ler, S.~K{\"u}chemann, M.~Bannert, D.~Dementieva, F.~Fischer, U.~Gasser, G.~Groh, S.~G{\"u}nnemann, E.~H{\"u}llermeier \emph{et~al.}, ``Chatgpt for good? on opportunities and challenges of large language models for education,'' \emph{Learning and individual differences}, vol. 103, p. 102274, 2023.

\bibitem{koh2022empirical}
H.~Y. Koh, J.~Ju, M.~Liu, and S.~Pan, ``An empirical survey on long document summarization: Datasets, models, and metrics,'' \emph{ACM computing surveys}, vol.~55, no.~8, pp. 1--35, 2022.

\bibitem{sorensen2022information}
T.~Sorensen, J.~Robinson, C.~Rytting, A.~Shaw, K.~Rogers, A.~Delorey, M.~Khalil, N.~Fulda, and D.~Wingate, ``An information-theoretic approach to prompt engineering without ground truth labels,'' in \emph{Proceedings of the 60th Annual Meeting of the Association for Computational Linguistics (Volume 1: Long Papers)}, 2022, pp. 819--862.

\bibitem{zhou2024quantifying}
Z.~Zhou, J.~Xiang, C.~Chen, and S.~Su, ``Quantifying and analyzing entity-level memorization in large language models,'' in \emph{Proceedings of the AAAI Conference on Artificial Intelligence}, vol.~38, no.~17, 2024, pp. 19\,741--19\,749.

\bibitem{wei2024jailbroken}
A.~Wei, N.~Haghtalab, and J.~Steinhardt, ``Jailbroken: How does llm safety training fail?'' \emph{Advances in Neural Information Processing Systems}, vol.~36, 2024.

\bibitem{gehman2020realtoxicityprompts}
S.~Gehman, S.~Gururangan, M.~Sap, Y.~Choi, and N.~A. Smith, ``{R}eal{T}oxicity{P}rompts: Evaluating neural toxic degeneration in language models,'' in \emph{Findings of the Association for Computational Linguistics: EMNLP 2020}, T.~Cohn, Y.~He, and Y.~Liu, Eds.\hskip 1em plus 0.5em minus 0.4em\relax Online: Association for Computational Linguistics, Nov. 2020, pp. 3356--3369.

\bibitem{kolla2024llm}
M.~Kolla, S.~Salunkhe, E.~Chandrasekharan, and K.~Saha, ``Llm-mod: Can large language models assist content moderation?'' in \emph{Extended Abstracts of the CHI Conference on Human Factors in Computing Systems}, 2024, pp. 1--8.

\bibitem{NEURIPS2023_a2cf225b}
N.~Scherrer, C.~Shi, A.~Feder, and D.~Blei, ``Evaluating the moral beliefs encoded in llms,'' in \emph{Advances in Neural Information Processing Systems}, vol.~36, 2023, pp. 51\,778--51\,809.

\bibitem{greshake2023not}
K.~Greshake, S.~Abdelnabi, S.~Mishra, C.~Endres, T.~Holz, and M.~Fritz, ``Not what you've signed up for: Compromising real-world llm-integrated applications with indirect prompt injection,'' in \emph{Proceedings of the 16th ACM Workshop on Artificial Intelligence and Security}, 2023, pp. 79--90.

\bibitem{liu2023optimizing}
Y.~Liu, H.~Du, D.~Niyato, J.~Kang, S.~Cui, X.~Shen, and P.~Zhang, ``Optimizing mobile-edge ai-generated everything (aigx) services by prompt engineering: Fundamental, framework, and case study,'' \emph{IEEE Network}, 2023.

\bibitem{dai2023safe}
J.~Dai, X.~Pan, R.~Sun, J.~Ji, X.~Xu, M.~Liu, Y.~Wang, and Y.~Yang, ``Safe {RLHF}: Safe reinforcement learning from human feedback,'' in \emph{The Twelfth International Conference on Learning Representations}, 2024.

\bibitem{armstrong2022using}
S.~Armstrong and R.~Gorman, ``Using gpt-eliezer against chatgpt jailbreaking,'' in \emph{AI ALIGNMENT FORUM}, 2022.

\bibitem{wang2023edge}
T.~Wang, Y.~Liang, X.~Shen, X.~Zheng, A.~Mahmood, and Q.~Z. Sheng, ``Edge computing and sensor-cloud: Overview, solutions, and directions,'' \emph{ACM Computing Surveys}, vol.~55, no. 13s, pp. 1--37, 2023.

\bibitem{bagheri2024ace}
S.~Bagheri, H.~Kermabon-Bobinnec, M.~E. Kabir, S.~Majumdar, L.~Wang, Y.~Jarraya, B.~Nour, and M.~Pourzandi, ``Ace-warp: A cost-effective approach to proactive and non-disruptive incident response in kubernetes clusters,'' \emph{IEEE Transactions on Information Forensics and Security}, 2024.

\bibitem{liang2024latency}
Y.~Liang, G.~Li, G.~Zhang, J.~Guo, Q.~Liu, J.~Zheng, and T.~Wang, ``Latency reduction in immersive systems through request scheduling with digital twin networks in collaborative edge computing,'' \emph{ACM Transactions on Sensor Networks}, 2024.

\bibitem{liang2023grouping}
Y.~Liang, M.~Yin, Y.~Zhang, W.~Wang, W.~Jia, and T.~Wang, ``Grouping reduces energy cost in directionally rechargeable wireless vehicular and sensor networks,'' \emph{IEEE Transactions on Vehicular Technology}, vol.~72, no.~8, pp. 10\,840--10\,851, 2023.

\bibitem{kaplan2020scaling}
J.~Kaplan, S.~McCandlish, T.~Henighan, T.~B. Brown, B.~Chess, R.~Child, S.~Gray, A.~Radford, J.~Wu, and D.~Amodei, ``Scaling laws for neural language models,'' \emph{arXiv preprint arXiv:2001.08361}, 2020.

\bibitem{harsanyi1967games}
J.~C. Harsanyi, ``Games with incomplete information played by “bayesian” players, i--iii part i. the basic model,'' \emph{Management science}, vol.~14, no.~3, pp. 159--182, 1967.

\bibitem{fudenberg1991game}
D.~Fudenberg and J.~Tirole, \emph{Game theory}.\hskip 1em plus 0.5em minus 0.4em\relax MIT press, 1991.

\bibitem{starrett1988foundations}
D.~A. Starrett, \emph{Foundations in Public Economics}.\hskip 1em plus 0.5em minus 0.4em\relax Cambridge University Press, 1988.

\bibitem{he2021debertav3}
P.~He, J.~Gao, and W.~Chen, ``De{BERT}av3: Improving de{BERT}a using {ELECTRA}-style pre-training with gradient-disentangled embedding sharing,'' in \emph{The Eleventh International Conference on Learning Representations}, 2023.

\bibitem{xue2024trojllm}
J.~Xue, M.~Zheng, T.~Hua, Y.~Shen, Y.~Liu, L.~B{\"o}l{\"o}ni, and Q.~Lou, ``Trojllm: A black-box trojan prompt attack on large language models,'' \emph{Advances in Neural Information Processing Systems}, vol.~36, 2024.

\bibitem{shi2022promptattack}
Y.~Shi, P.~Li, C.~Yin, Z.~Han, L.~Zhou, and Z.~Liu, ``Promptattack: Prompt-based attack for language models via gradient search,'' in \emph{CCF International Conference on Natural Language Processing and Chinese Computing}.\hskip 1em plus 0.5em minus 0.4em\relax Springer, 2022, pp. 682--693.

\bibitem{du2022ppt}
W.~Du, Y.~Zhao, B.~Li, G.~Liu, and S.~Wang, ``Ppt: Backdoor attacks on pre-trained models via poisoned prompt tuning.'' in \emph{Proceedings of the Thirty-First International Joint Conference on Artificial Intelligence}, 2022, pp. 680--686.

\bibitem{liu2024formalizing}
Y.~Liu, Y.~Jia, R.~Geng, J.~Jia, and N.~Z. Gong, ``Formalizing and benchmarking prompt injection attacks and defenses,'' in \emph{33rd USENIX Security Symposium (USENIX Security 24)}, 2024, pp. 1831--1847.

\bibitem{hines2024defending}
K.~Hines, G.~Lopez, M.~Hall, F.~Zarfati, Y.~Zunger, and E.~Kiciman, ``Defending against indirect prompt injection attacks with spotlighting,'' \emph{arXiv preprint arXiv:2403.14720}, 2024.

\bibitem{wang2024a}
Z.~Wang and Y.~Qi, ``A closer look at adversarial suffix learning for jailbreaking {LLM}s,'' in \emph{ICLR 2024 Workshop on Secure and Trustworthy Large Language Models}, 2024.

\bibitem{caselli2020hatebert}
T.~Caselli, V.~Basile, J.~Mitrovi{\'c}, and M.~Granitzer, ``Hatebert: Retraining bert for abusive language detection in english,'' in \emph{Proceedings of the 5th Workshop on Online Abuse and Harms (WOAH 2021)}, 2021, pp. 17--25.

\bibitem{mabrouk2023intrusion}
A.~Mabrouk and A.~Naja, ``Intrusion detection game for ubiquitous security in vehicular networks: a signaling game based approach,'' \emph{Computer Networks}, vol. 225, p. 109649, 2023.

\bibitem{ge2023gazeta}
Y.~Ge and Q.~Zhu, ``Gazeta: Game-theoretic zero-trust authentication for defense against lateral movement in 5g iot networks,'' \emph{IEEE Transactions on Information Forensics and Security}, 2023.

\end{thebibliography}

\end{document}